\documentclass[conference]{IEEEtran}
\IEEEoverridecommandlockouts
\usepackage{cite}
\usepackage{amsmath,amssymb,amsfonts}
\usepackage{algorithmic}
\usepackage{graphicx}
\usepackage{textcomp}
\usepackage{tikz}
\usepackage{xcolor}
\usepackage{amsthm}
\usepackage[english]{babel}
\newtheorem{theorem}{Theorem}
\newtheorem{proposition}{Proposition}
\theoremstyle{lemma}
\newtheorem{lemma}{Lemma}
\theoremstyle{definition}
\newtheorem{definition}{Definition}
\def\BibTeX{{\rm B\kern-.05em{\sc i\kern-.025em b}\kern-.08em
    T\kern-.1667em\lower.7ex\hbox{E}\kern-.125emX}}

\begin{document}

\title{Entropic Weighted Rank Function\\
}

\author{
\IEEEauthorblockN{ Mohammad Rashid$^1$, {Elahe Ghasemi}$^1$, {Javad B. Ebrahimi}$^{1,2}$}
\IEEEauthorblockA{\textit{$^1$ Department of Mathematical Sciences, Sharif University of Technology, Tehran, Iran} \\
\textit{$^2$ IPM, Institute for Research in Fundamental Sciences, Tehran, Iran}\\
Email: \{m.rashid95, elahe.ghasemi96\}@student.sharif.edu, javad.ebrahimi@sharif.ir }
}

\maketitle

\begin{abstract}
It is known that the entropy function over a set of jointly distributed random variables is a submodular set function. However, not any submodular function is of this form. In this paper, we consider a family of submodular set functions, called weighted rank functions of matroids, and study the necessary or sufficient conditions under which they are entropic.  We prove that weighted rank functions are located on the boundary of the submodularity cone. For the representable matroids over a characteristic 2 field, we show that the integer valued weighted rank functions are entropic. We derive a necessary condition for constant weight rank functions to be entropic and show that for the case of graphic matroids, this condition is indeed sufficient. Since these functions generalize the rank of a matroid, our findings generalize some of the results of Abbe et. al. in \cite{abbe2019entropic} about entropic properties of the rank function of matroids.
\end{abstract}


\section*{Introduction}
For a set of random variables with a given joint distribution, the entropy region is defined as the collection of the entropies of the joint of non-empty subsets of those variables. Entropic region plays an important role in information theory. A geometric representation of the entropy region is the set of points in $\mathbb{R}^{2^n-1}$ where $n$ is the number of random variables.
Consider a point in $\mathbb{R}^{2^n-1}$ with its coordinates indexed by the non-empty subsets of the variables and the value of the $S$-th position is the joint entropy of the variables in the set $S$. We refer to any such point as an entropic point. The set of all entropic points is called the entropic region and is denoted by $\Gamma^*_n$. It is known that the closure of $\Gamma^*_n$ is a convex cone. However, we do not exactly know the structure of the boundary of this set. Another important property of $\Gamma^*_n$ is the submodularity; (See \cite{yeung2008information}.) Any inequality that is obtained by taking a non-negative linear combination of submodularity constraints is called a Shannon type inequality. Let $\Gamma_n$ be the set of all the non-negative points in $\mathbb{R}^{2^n-1}$ which satisfy all the submodularity constraints and is also non-decreasing. $\Gamma_n$ is known to be a convex cone. (\cite{yeung2008information}.) $\Gamma_n$ is called the submodularity cone. Therefore, the submodularity constraint of the entropic point can be written as $\Gamma^*_n \subseteq \Gamma_n$. In \cite{yeung2008information}, it is proven that for $n>3$, the above containment is strict. 
There has been a significant effort to derive inequalities beyond Shannon type ones, that are satisfied for any entropic point. These inequalities, called non-Shannon type, are necessary conditions for a submodular function to be entropic. In this work, we take a different path and instead of finding necessary conditions, we investigate sufficient conditions for a class of submodular functions, called weighted ranks, to be entropic. The reason we are interested in weighted rank functions, besides the fact that they contain the family of rank function of matroids, is that they have nice geometric interpretation which we explain more in Theorem \ref{thm:vertex}.

In \cite{abbe2019entropic}, it is shown that if $\mathcal{M}$ is a matroid on a ground set of size $n$, then the rank function of $\mathcal{M}$ is entropic when the matroid is representable over a characteristic 2 finite field. Our founding is an extension of this result since, for these matroids, we prove that weighted rank functions are entropic and the rank function is a special case of weighted rank functions. We must emphasise that, as we show an example in Figure \ref{fig:fig2} a simple reduction of weighted rank function to the rank function of a new matroid does not work. 

Beyond integral weights, we also consider the case where all the weights are equal non negative numbers. We derive necessary as well as sufficient conditions for the general weighted rank function of binary matroids to be entropic. For graphic matroids, we provide a necessary and sufficient condition for a constant weight rank function to be entropic.

In the rest of this paper, after we briefly review the basic definitions and notations, we state the problem of interest and the main results. Then we present the technical proofs and finally, we discuss two important points about our results. 

\section*{Preliminaries}
In this section we review some standard definitions and basic facts related to convex polytopes \cite{brondsted2012introduction}, matroids \cite{oxley2006matroid} and information theory\cite{cover1991information , yeung2008information}.

\subsection{Matroid Theory}
A \textit{matroid} $\mathcal{M} = (M, \mathcal{I})$ is defined as a pair of a finite set $E$ and a family $\mathcal{I} \subseteq 2^M$ satisfying the following conditions:
\begin{itemize}
    \item $\emptyset \in \mathcal{I}$
    \item $I_1 \subseteq I_2 \in \mathcal{I}$ implies $I_1 \in \mathcal{I}$
    \item if $I_1 , I_2 \in \mathcal{I}$ and $|I_1| \le |I_2|$ then there is an element $e \in I_2 \backslash I_1$, such that $I_1 \cup \{e\} \in \mathcal{I}$
\end{itemize}
The  $M$ is called the \textit{ground set} and each $I \in \mathcal{I}$ is called an \textit{independent set}.

The matroid \textit{rank function} denoted by $r(.)$ is defined as:
$$r(I) = \max \{ |I'|: I' \subseteq I, I \in \mathcal{I} \}$$
where $|I|$ is cardinality of $I$. Matroid rank function enjoys the following properties.
\begin{itemize}
\item $r(A)\in \mathbb{Z}^{\geq 0}$
\item $r(A)\leq 1$ when $A$ is a singleton.
\item $r(.)$ is a non-decreasing function
\item $r(.)$ is a submodular function; that is, $r(A)+r(B)\geq r(A\cup B) + r(A \cap B)$ for all $A,B\subseteq M$.
\end{itemize}

Furthermore, any function which satisfies the above properties is the rank function of some matroid. (See \cite{oxley2006matroid}).

A \textit{circuit} of $\mathcal{M}$ is a nonempty minimial dependent set
and a \textit{base} of $\mathcal{M}$ ,  $\mathcal{B}(\mathcal{M})$,
 is a maximal independent set.

\begin{proposition}
\label{prop:circuit}
If $C_1$ and $C_2$ are distinct circuits and $e \in C_1 \cap C_2$ , then $(C_1 \cup C _2) \backslash e$ contains a circuit.
\end{proposition}

A graphic matroid  is a matroid whose ground set is the edge set of a graph and its independent sets are the edges which form a forest. Another important class of matroids is the class of linear matroids that is defined as follows. Let $M=\{1,2,\ldots,n\}$ and $A$ be a matrix with columns $v_1,\ldots , v_n$ over $\mathbb{F}$. Define:
$$ \mathcal{I} = \{ I \subseteq M | \{ v_i | i \in I \} \text{ are linearly independent vectors} \}.$$
Then, $\mathcal{M}=(M, \mathcal{I})$ is a linear matroid over the field $F$. We denote such a matroid by $\mathcal{M} = \mathcal{M}[A]$ and call $A$ a representative matrix of the matroid. For convenient, we take $I \subseteq M$ to be a set of vectors corresponding to the elements of $I$.

A matroid is representable over a field $\mathbb{F}$ if it is isomorphic to a linear matroid over the field $\mathbb{F}$. 

All graphic matroids are representable over $\mathbb{F}_2$ and $\mathbb{R}$ and their representative matrices are the incidence matrix of the graph.

\subsection{Convex Polytopes}
Next, we review basics of convex polytopes. We follow the. terminology of \cite{brondsted2012introduction}. More results and examples on this topic can be found in \cite{schrijver2003combinatorial}.

A set $C$ is convex when the closed segment between any two
points of $C$ is contained in $C$.  $C$ is called affine when the line through any two distinct points in it lies entirely inside $C$.

A \textit{polyhedron} in $\mathbb{R}^n$ is the intersection of finitely many half-spaces. A \textit{polytope} is a bounded polyhedron. Equivalently, a polytope is the smallest convex set containing a given finite set of points called its. \textit{vertices}. (See \cite{brondsted2012introduction} for a proof.)

Given a convex set $P \subseteq \mathbb{R}^n$ , a point $x \in P$ is an extreme point of P if there do not exist points $u, v \neq x$ in $P$ such that $x$ is a convex combination of $u$ and $v$. For polytopes, both the notions of vertex points and extreme points coincide. (See \cite{brondsted2012introduction} for more details.)

A \textit{cone} is a set $C \subseteq \mathbb{R}^n$ such that it is closed under non-negative scalar multiplication. A \textit{convex cone} is a cone that is also convex.

\subsection{Information theory}
For a random variable $X$ the \textit{entropy} $H(X)$  is defined as:\\
\begin{center}
    $H(X) = - \sum\limits_{x} p(x) \log p(x)$
\end{center}
\begin{definition}
We define $Ber(p)$ to be a Bernoulli variable that takes value 1 with probability $p$ and the value $0$ with probability $q = 1 - p$. 
\end{definition}
When $X$ is distributed as $Ber(\frac{1}{2})$ we have $H(X)=1$.
Let $H_n$ be the $k$-dimensional Euclidean space with the coordinates labeled
by $h_\alpha$ , $\alpha \in 2^{\mathcal{N}_n} \backslash \{\emptyset\}$, where $\mathcal{N}=\{1,2,\ldots,n\}$ and $h_\alpha$ corresponds to the value of $H_\theta(\alpha)$ for any
collection $\theta$ of $n$ random variables. We will refer to $H_n$ as the entropy space
for $n$ random variables. Then an entropy function $H_\theta$ can be represented by a column vector in $H_n$.\\
A column vector $h \in H_n$ is called
entropic if $h$ is equal to the entropy function $H_\theta$ of some collection $\theta$ of $n$ random variables.(\cite{yeung2008information})

\section*{Problem Formulation and Results}
Throughout this section, we assume that $\mathcal{M}=(M,\mathcal{I})$ is a representable matroid over the finite field $\mathbb{F}$ and the ground set $M=\{1,2,\ldots , n\}$. We let $w:M\to \mathbb{R}^{\geq 0}$ be a weight function defined on $M$. 
Given $w$, we define the function $\phi_w : 2^M \to \mathbb{R}$  as follows:\[\phi_{\mathcal{M},w}(A):= \max_{I\subseteq A,I\in \mathcal{I}} {\sum_{a\in I}{w(a)}}\]
We will use $\phi_{w}$ for $\phi_{\mathcal{M},w}$ when $\mathcal{M}$ is clear from the context. For an arbitrary weight function $w$, $\phi_w$ is called the weighted rank of $\mathcal{M}$ with respect to $w$.
\begin{proposition}
\label{pro:submodular}
For every non-negative weight function $w$, $\phi_w$ is a submodular function.
\end{proposition}
In Appendix \ref{proof:submodularphi}, we present our proof of the fact that  for any non-negative weight $w$, $\phi_w$ is submodular. Therefore, if we consider the vector $\Phi=(\phi (A))_{\emptyset\neq A\in 2^{M}}$, it belongs to $\Gamma_n$, the convex submodularity cone of order $n$.  In the next section we prove that $\Phi$ is in fact a vertex of the polytope $P=\Gamma_n \cap S$ in which $S$ is an Affine subspace of co-dimension $n$ containing all the points in $\mathbb{R}^{2^n-1}$ whose coordinates $\{a\}$ is equal to $w(a)$ for every single point $a\in M$. In other words, if we consider all the points in $\mathbb{R}^{2^n-1}$ that share the same values $w(a)$ at the coordinates indexed by the subsets $\{a\}$. This is summarised in the following theorem.
\begin{theorem}
\label{thm:vertex}
For every choice of non-negative weights $w$ and any matroid $\mathcal{M}$, $\phi_{\mathcal{M},w}$ is a vertex of the polytope $P$ defined as above.
\end{theorem}

In particular, this implies that $\Phi$ belongs to the boundary of $\Gamma_n$. \\
Now, the question of interest is that, for a given matroid $\mathcal{M}$ and the weight function $w$, is $\Phi$ an entropic point;  i.e. is $\Phi \in \Gamma^*_n$? \\

In \cite{abbe2019entropic}, Abbe et.al. proved that for any representable matroid over the binary field and the constant weight function $w=1$,  the answer is affirmative. In fact, for the constant function $w=1$,  $\phi_w$ is simply the usual rank function of the matroid. In this paper we prove the following results.

\begin{theorem}
\label{thm:main}
Let $\mathcal{M}$ be a representable matroid over the field $\mathbb{F}_2$ and $w$ is any non-negative integer weight function on $M$. Then, $\phi_{w, \mathcal{M}}$ is entropic. \end{theorem}

We also consider the non-integral weight function over certain matroids and derive necessary or sufficient conditions under which $\Phi$ is entropic.

\begin{theorem}
\label{thm:logk}
$\mathcal{M} = (M , \mathcal{I})$ is a matroid with at least one circuit and $w$ is constant weight function on $M$. Then, $\phi_w$ is entropic if $w = \log k$ for $k \in \mathbb{N}$.
\end{theorem}
Finally, we prove the converse of Theorem \ref{thm:logk} for graphic matroids.
\begin{theorem}
\label{thm:inverse}
 $\mathcal{M} = (M , \mathcal{I})$ is a graphic matroid of weight function $w = \log k$ for $k \in \mathbb{N}$. $\phi_w$ is entropic.
\end{theorem}

\section*{Technical Details of the Proofs}
In this section after we stated some useful lemmas, we prove the main results of the paper stated in the previous section. 

\begin{lemma}
\label{lemma:MWB}
Let $\mathcal{M}$ be a matroid on the ground set $M$. Suppose that we start from $M$, sequentially pick a circuit, and remove the element with minimum weight from it. Then, the resulting subset  is a maximum weighted base $B$.
\end{lemma}
\begin{proof}
Since there is no circuit in $B$, it is an independent set. Also, at each round, the remaining elements are generating since we always remove the element from a circuit. 
Take any other base $B'$. Let $e$ be an element that $e \in B \backslash B'$. Adding $e$ to $B'$ creates a circuit, and deleting any element of that circuit would create another base. Let's add $e$ and delete the lightest element of that circuit. That lightest element is definitely not $e$, because $e$ is not the lightest element of any circuit. Thus, we added $e$ to $B'$, and then deleted an element lighter than $e$. This means that we have increased the total weight of $B'$. Therefore, $B'$ is not a maximum weighted base. Since some maximum weighted base must exist, it can only be $B$.
\end{proof}
An immediate corollary of Lemma \ref{lemma:MWB} is that the maximum weighted base of a matroid, after removing the minimum weight element of a circuit in the matroid, remains unchanged. We are now ready to prove Theorem \ref{thm:vertex}. 

\begin{proof}
Fix a weight function $w$ on $[n]$. Let $P \subseteq \mathbb{R}^{2n-1}$ be the affine space of all the points $v$ such that $v_{{i}}=w_i$. One can easily observe that $P$ is convex. Let $S := P \cap \Gamma$. Since $\Gamma_n$ is a convex cone and $P$ is convex, $S$ is a convex polyhedron. Also, it is trivial that the submodulartiy implies that for any point $v \in \Gamma_n$, $v$ is contained in the hypercube whose $A$-th coordinate is between $0$ and $\sum\limits_{i \leq n}^{} w_i$. In other words, for any point in the submodularity cone $\Gamma_n$, the first $n$-coordinates impose an upper bound on any other coordinate. Thus, $S$ is contained in a box with the diagonal $(0)_{A \subseteq [n]}$ and $(\sum\limits_{i \in A} w(i))_{A \subseteq [n]}$. Therefore, $S$ is indeed a polytope. Therefore, in order to prove that a point is a vertex of $S$, it is enough to show that it is an extreme point.


Suppose there is a matroid $\mathcal{M}$ such that $\phi_{\mathcal{M}, w}$ is not a vertex of the polytope $S$ or equivalently, there are two functions $f_w, q_w$ in $S$ such that for every subset $X \subseteq [n]$:
$$\alpha f_w(X) + (1 - \alpha) g_w(X) = \phi_{\mathcal{M},w}(X)$$
With induction on size of $X$ we prove that $$ f_w(X) = g_w(X) = \phi_{\mathcal{M}, w}(X)$$
Since $f_w$ and $g_w$ are in $P$,
for each $i \in [n]$ we have $$ f_w(\{i\}) = g_w(\{i\}) = \phi(\{i\}) $$
Now, assuming the claim for every subset $X$ of size of at most $k$, if $X \subseteq [n]$ and $|X| \leq k$ we have:
$$f_w(X) = g_w(X) = \phi_{\mathcal{M}, w}(X)$$
Let $Y$ be of size $k+1$, and  $j$ be the smallest weight in $Y$. We know that:
$$ f_w(Y\backslash j) = g_w(Y\backslash j) = \phi_{\mathcal{M}, w} (Y \backslash j)$$
Consider the following cases:
\begin{itemize}
\item[1)]
$r(Y \backslash j) < r(Y)$. In this case, $$\phi_w(Y) = \phi_w(Y \backslash j) + \phi_w(\{j\}) = \phi_w(Y \backslash j) + w_j.$$
On the other hand, from the submodularity conditions we have:
$$ f_w(Y) \leq f_w(Y \backslash j) + f_w(j) = \phi_w (Y \backslash j) + w_j $$
$$ g_w(Y) \leq g_w (Y \backslash j) + g_w(j) = \phi_w (Y \backslash j) + w_j
$$
Since $\phi$ is a convex combination of $f_w$ and $g_w$, therefore $f_w(Y)$ and $g_w(Y)$ must be equal to $\phi_w(Y)$.

\item[2)] $r(Y \backslash j)=r(Y)$. Since $j$ has the smallest weight in $Y$ and $j$ belongs to some circuit in $Y$, by Lemma \ref{lemma:MWB} we have $\phi_{\mathcal{M}, w}(Y) = \phi_{\mathcal{M}, w}(Y \backslash j)$.\\

Also, from monotonicity of $f_w$ and $g_w$ we have:
$$ f_w(Y) \geq f_w(Y \backslash j) = \phi_w(Y \backslash j) = \phi_w(Y) $$
and similarly $g_w(Y) \geq \phi_w(Y)$. Thus, both $f_w(Y)$ and $g_w(Y)$ must be equal to $\phi_w(Y)$ and this proves the claim.
\end{itemize}
\end{proof}
Let $\mathcal{M = (M , \mathcal{I})}$ be a binary matroid with representative matrix $$ A = 
\begin{bmatrix}
    v_1 & v_2 & . & . & .  & v_n \\
\end{bmatrix}
$$ where $|M| = n$ and $v_i$ is a $m \times 1$ vector.
\\ Define $w_{\text{max}} = \text{max} \{ w_e \text{ }|\text{ }  e \in M \}$.
We define the matrix $X$ of variables as:
$$ X = 
\begin{bmatrix}
    X_1^1 & X_1^2 & . & . & .  & X_1^m \\
    . \\
    .\\
    X_{w_{\text{max}}}^1 & X_{w_{\text{max}}}^2 & . & . & .  & X_{w_{\text{max}}}^m \\
\end{bmatrix}
$$

$Y_e$ is the variable corresponding to the element $e$ of weight $w_e$. It is the join of the variables of the first $w_e$ variables in the matrix $X \times v_e$ such that each $X_i^k$ is an independent $Ber(\frac{1}{2})$. If we define $X_i$ to be the $i$-th row of the matrix $X$, then: 
\begin{equation}
\label{equ:Y_e}
    Y_e = (X_1 v_e, X_2 v_e, ..., X_{w_e} v_e)
\end{equation}
\begin{lemma} 
\label{lemma:removeelement}
Let $\mathcal{M} = (M, \mathcal{I})$ be a matroid over $\mathbb{F}_2$ with the random variable $Y_e$ assigned to each $e \in M$ as in Equation \ref{equ:Y_e} and $e'$ be the lightest element in a circuit $C$ in $\mathcal{M}$. If $A\subseteq M$ containing $C$ then we have:
$$
H(\bigcup\limits_{e\in A}Y_e)=\bigcup\limits_{e\in A, e\neq e'}Y_e
$$ 
\end{lemma}
\begin{proof}
Define $Y_i^{t}$ to be:
$$ Y_i^{t} = \left\{
\begin{array}{ll}
      (X_1 v_i, X_2 v_i..., X_{t} v_i) & \text{if } t \leq w_i \\
      Y_i & \text{otherwise} \\
\end{array} \right. $$
Let $C= \{ v_{i_1} , v_{i_2} , ... , v_{i_k} \}$ be a circuit where the weight of $v_{i_j}$ is $w_{i_j}$ and suppose that $w_{i_1} \geq w_{i_2} \geq ... \geq w_{i_k}$.\\
Since $C$ is a circuit
$ \sum\limits_{j=1}^{k-1}  v_{i_j} = v_{i_k} $ and we have that $$\forall l \leq w_{i_k} \text{ : } \sum\limits_{j = 1}^{k-1}  X_l v_{i_j} = X_l v_{i_k}$$.
\\
Since for every $j$ we have that $w_{i_j} \geq w_{i_k}$:
$$ \sum\limits_{j=1}^{k-1}  Y_{i_j}^{w_{i_k}} = Y_{i_k}^{w_{i_k}} = Y_{i_k} $$
$$ \Longrightarrow H(Y_{i_k} | Y_{i_1}^{w_{i_k}}, Y_{i_2}^{w_{i_k}}, ..., Y_{i_{k-1}}^{w_{i_k}}) = 0 $$
$$ \Longrightarrow H(Y_{i_k} | Y_{i_1}, Y_{i_2}, ..., Y_{i_{k-1}}) = 0$$
$$\Longrightarrow H(Y_{i_k} | \bigcup\limits_{j=1, j \neq i_k}^{n} Y_{j}) = 0$$
Therefore, if we remove $v_{i_k}$ from the matroid, its entropy remains the same.
\end{proof}
To prove Theorem \ref{thm:main}, we need one more technical lemma about the statistical independence of linearly independent combinations of statistically independent Bernoulli random variables with parameter $\frac{1}{2}$. The following proposition, proved in \cite{nyberg2017statistical}, (Page 2, Theorem 1) is particularly useful.

\begin{proposition}
\label{prop:independence}
Linearly independent combinations of statistically independent Bernoulli random variables are statistically independent.
\end{proposition} 

\begin{lemma}
\label{lemma:summation} Let $\mathcal{M} = (M, \mathcal{I})$ be a representable matroid over $\mathbb{F}_2$ with the random variable $Y_e$ assigned to each $e \in M$ as in Equation \ref{equ:Y_e}. The entropy of an independent set $I \in \mathcal{I}$ is the summation of entropies of all the elements of $I$.
\end{lemma}

\begin{proof}
First of all, notice that the summation of two independent Bernoulli random variables with parameter $\frac{1}{2}$ is again a $Ber(\frac{1}{2})$ random variables. Thus, to complete the proof, all we need to justify is that the random variables assigned to the elements are independent. 

In fact, assuming the independency condition, the entropy function of the random variables becomes the summation of the entropies of the individual ones. From this, the entropy of each element is equal to the summation of the entropies of the same number of Bernoulli random variables with parameter $\frac{1}{2}$ as its weight. Also, the entropy of the random variables assigned to the elements in $I$ can be expanded as the summation of the entropies of the variables in each element, again under the independence condition.

Therefore, we just have to prove that the set of all the random variables assigned to the elements of an independent set $I$ are statistically independent. This is also a direct implication of proposition \ref{prop:independence}.

\end{proof}

\begin{figure}
    \centering

\tikzset{every picture/.style={line width=0.75pt}} 

\begin{tikzpicture}[x=0.5pt,y=0.5pt,yscale=-1,xscale=1]

\fill [black] (283.09,188.4) circle (2);
\fill [black] (429.8,319.4) circle (2);
\fill [black] (141,319.4) circle (2);
\fill [black] (497,190.05) circle (2);
\fill [black] (497.5,320.55) circle (2);

\draw   (283.09,188.4) -- (429.8,319.4) -- (141,319.4) -- cycle ;
\draw    (497,190.05) -- (283.09,188.4) ;
\draw    (497.5,320.55) -- (497,190.05) ;

\draw (286,170) node [anchor=north west][inner sep=0.75pt]   [align=left] {$v_3$};
\draw (121,298) node [anchor=north west][inner sep=0.75pt]   [align=left] {$v_5$};
\draw (434,303) node [anchor=north west][inner sep=0.75pt]   [align=left] {$v_4$};
\draw (105,202) node [anchor=north west][inner sep=0.75pt]   [align=left] {$ 1, (X_1^{v_3} + X_1^{v_5}) $};
\draw (170,340) node [anchor=north west][inner sep=0.75pt]   [align=left] {$ 2, (X_1^{v_4} + X_1^{v_5}, X_2^{v_4} + X_2^{v_5}) $};
\draw (338,215) node [anchor=north west][inner sep=0.75pt]   [align=left] {$ 1, (X_1^{v_3} + X_1^{v_4}) $};
\draw (484.5,322) node [anchor=north west][inner sep=0.75pt]   [align=left] {$v_1$};
\draw (497,173) node [anchor=north west][inner sep=0.75pt]   [align=left] {$v_2$};
\draw (330.5,160.5) node [anchor=north west][inner sep=0.75pt]   [align=left] {$ 1, (X_1^{v_2} + X_1^{v_3}) $};
\draw (501,249.5) node [anchor=north west][inner sep=0.75pt]   [align=left] {$ 1, (X_1^{v_1} + X_1^{v_2}) $};

\end{tikzpicture}
    \caption{The construction of variables for graphic matroids case; each edge is shown by its weight and the variable corresponding to it.}
    \label{fig:variablesfigure}
\end{figure}
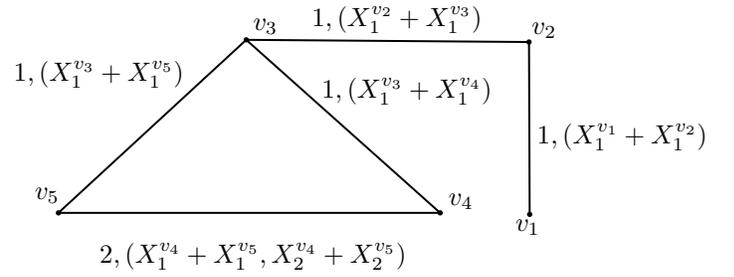

\begin{proof}[Proof (Theorem \ref{thm:main} graphic matroids)]
Let $G = (V , E)$ be a graphic matroid.
We define variable $Y_e$ corresponding to edge $e=(v , u) \in E$ of weight $w_e$ to be $$ Y_e = (X_1^v + X_1^u , X_2^v + X_2^u , ... , X_{w_e}^v + X_{w_e}^u)$$ That each $X_i^k$ is an independent $Ber(\frac{1}{2})$. All forests in a graphic matroid are independent sets. Therefore, a base in a graphic matroid is a spanning forest.
Given a subset $G' = (V ,  E' \subseteq E)$, by Lemma \ref{lemma:removeelement} and Lemma \ref{lemma:MWB} the entropy of the variables defined on $E'$ is equal to the entropy of the variables defined on the edges of the maximum weighted forest $F$. By Lemma \ref{lemma:summation} we have that:
$$ H(\bigcup_{e \in F} Y_e) = \sum\limits_{e \in F} H(Y_e) $$

For each i, $X_{i}^{k}$ is independent from others. So $X_{i}^{v} + X_{i}^{u}$ is also independent from others. By definition of entropies for each $e = (u, v) \in E$ we have that:
$$H(X_i^v + X_i^u) = 1$$
$$H(Y_e) = H(X_1^v + X_1^u , X_2^v + X_2^u , ... , X_{w_e}^v + X_{w_e}^u) = w_e$$
$$ \Longrightarrow H(\bigcup_{e \in E'} Y_e) = H(\bigcup_{e \in F} Y_e) = \sum\limits_{e \in F} w_e = \phi_w(E')$$

\end{proof}
The special case above is intended to give intuition about the proof for the general case. Below, we prove the theorem in the general setting.
\begin{proof}[Proof (Theorem \ref{thm:main} general case)]
Given a subset $I \subseteq M$, by Lemma \ref{lemma:removeelement} and Lemma \ref{lemma:MWB} the entropy of the variables defined on the elements of $I$ is equal to the entropy of the variables defined on the elements of the maximum weighted independent set $B \subseteq I$.\\
Since $X_1, X_2, ..., X_w$ are independent, multiplying them by a non zero vector $v_e$ does not change the independence. Therefore, $X_1 v_e, X_2 v_e, ..., X_{w_e} v_e$ are independent. $X_j v_e$ for each $j$ is the summation of some independent $Ber(\frac{1}{2})$s. Therefore, $H(X_j v_e) = 1$.

\begin{equation}
    \label{equ:1}
    H(Y_e) = H(\bigcup\limits_{j=1}^{w_e} X_j v_e) = \sum\limits_{j = 1}^{w_e} H(X_j v_e) = w_e
\end{equation}

By Equation \ref{equ:1} and Lemma \ref{lemma:summation} we have that:
$$ H(\bigcup_{e \in I} Y_e) = H(\bigcup_{e \in B} Y_e) =  \sum\limits_{e \in B} H(Y_e) =\sum\limits_{e \in B} w_e = \phi_w(I)$$
\end{proof}

\begin{proof}[Proof (Theorem \ref{thm:logk})]
Suppose $\phi_{w, \mathcal{M}}$ is entropic where $w = w_0$ is a constant weight function.
Take $C = \{ e_1, e_2, ..., e_m \}$ to be a circuit in $\mathcal{M}$. Since $\phi_{w, \mathcal{M}}$ is entropic there exist variables $X_1, ..., X_m$ such that:
$$ H(\bigcup\limits_{i \in I} X_{i}) = \phi_{w}(I) \text{ for } I \subseteq C $$

Let $\mathcal{X}_i$ be the support of $X_i$. Since $w = w_0$, 
$$ H(\bigcup\limits_{i \in I} X_{i}) = \left\{
\begin{array}{ll}
      (m-1) w_0 & \text{if } |I| = m \\
      |I| w_0 & \text{if } |I| < m \\
\end{array} \right. $$

$$ \Longrightarrow H(\bigcup\limits_{i \in {C \backslash {e_j}}} X_{i}) = (m-1) w_0 = H(\bigcup\limits_{i \in {C}} X_{i})$$ $$ \Longrightarrow H(X_j | \bigcup\limits_{i \in {C \backslash {e_j}}} X_{i}) = 0$$
Therefore, if we the values of $X_1, ..., X_{m-1}$ are known,  the value of $X_m$ is uniquely determined. Thus, if $X_i=\alpha_i$   for arbitrary values of $\alpha_i \in \chi_i$ for $1 \leq i\leq m-1$, then:
$$ p(X_m = \beta | X_1 = \alpha_1 , X_2 = \alpha_2, ..., X_{m-1} = \alpha_{m-1}) \in \{0,1\}$$
Since $C$ is a circuit, each set $S = \{ X_i | i \in L \subseteq [m] \text{ and } |L| \leq m-1 \}$ is independent. Therefore, 
$$ \alpha_i \in \mathcal{X}_i \Longrightarrow p(X_i = \alpha_i) > 0 \Longrightarrow$$
$$ p(X_1 = \alpha_1 , X_2 = \alpha_2 , ..., X_{m-1} = \alpha_{m-1}) =$$ $$ p(X_1 = \alpha_1) p(X_2 = \alpha_2) ... p(X_{m-1} = \alpha_{m-1}) > 0$$
And because $H(X_m | X_1, X_2,...,X_{m-1}) = 0$ there exists $\alpha_m \in \mathcal{X}_m$ such that 
$$p(X_m = \alpha_m | X_1 = \alpha_1, ..., X_{m-1} = \alpha_{m-1})=1$$
$$ \Longrightarrow p(X_1 = \alpha_1,...,X_m = \alpha_m) =$$
$$ p(X_m = \alpha_m | \bigcup\limits_{j=1}^{m-1}
(X_j = \alpha_j)) \times p( \bigcup\limits_{j=1}^{m-1}
(X_j = \alpha_j))$$
$$ = p(X_1 = \alpha_1) p(X_2 = \alpha_2) ... p(X_{m-1} = \alpha_{m-1})$$
Also, for each $i$:
$$ p(X_i = \alpha_i | X_j = \alpha_j \text{ for } 1 \leq j \leq m \text{ and } j \neq i )=$$
$$\frac{p(X_1 = \alpha_1 , ..., X_m = \alpha_m)}{
\prod\limits_{j=1 \text{ and } j \neq i}^{m}
p(X_j = \alpha_j) 
} > 0 $$
$$ \Longrightarrow  p(X_i = \alpha_i | X_j = \alpha_j \text{ for } 1 \leq j \leq m \text{ and } j \neq i ) = 1 $$
Thus, we also have:
$$ p(X_1 = \alpha_1 , ..., X_m = \alpha_m) = \prod\limits_{j=1 \text{ and } j \neq i}^{m}
p(X_j = \alpha_j)  $$
$$ \Longrightarrow \forall i \text{ : } p(X_i = \alpha_i) = p(X_m = \alpha_m) $$
$$ \Longrightarrow  \forall i,j < m \text{ : } p(X_i = \alpha_i) = p(X_j = \alpha_j) $$
Since we chose $\alpha_i$ and $ \alpha_j$ arbitrarily, we proved that:
$$ \forall \alpha_{i_1}, \alpha_{i_2} \in \mathcal{X}_i \text{ and } \forall \alpha_j \in \mathcal{X}_j$$
$$ \Longrightarrow p(X_i = \alpha_{i_1}) = p(X_i = \alpha_{i_2}) = p(X_j = \alpha_{j}) $$
$$ \Longrightarrow p(X_i = \alpha_i) = \frac{1}{|\mathcal{X}_i|} = p(X_j = \alpha_j)$$
Therefore, $\forall i$ : $X_i$ is uniform and $|\mathcal{X}_1| = |\mathcal{X}_2| = ... = |\mathcal{X}_m|$.
Then $H(X_i) = \log |\mathcal{X}_i|$ for all $i$. Therefore, $w = \log |\mathcal{X}_1|$.
\\
\\
\end{proof}
\begin{proof}[Proof (Theorem \ref{thm:inverse})]
Let $G = (V , E)$ be a graph of weight function $w$ on $E$. Take a random orientation on the edges that gives an incidence matrix $A$. For each directed edge $e = (u, v)$, define the variable corresponding to $e$ to be $Y_e = X_v - X_u$ where for $\forall i \in V$ each $X_i$ is a uniform random variable  on $[k]$. Then $Y_e$ is a uniform random variable on $[k]$ as well. 

Take $$ A = 
\begin{bmatrix}
    v_1 & v_2 & . & . & .  & v_m \\
\end{bmatrix}$$ where each column corresponds to an edge of $E$.
Let $C = (v_{i_1}, v_{i_2}, ..., v_{i_l})$ be a cycle in $G$. Then there exists $\alpha_j \in \{-1,1\}, \forall 1 \leq j \leq l$ that $ \sum\limits_{j = 1}^{l} \alpha_j v_{i_j} = 0 $. Take an arbitrary edge $e$ then
$ v_e = \frac{-1}{\alpha_e} \sum\limits_{j = 1, j \neq e}^{l} \alpha_j v_{i_j}$.
Let $X$ be the vector of variables
$X = \begin{bmatrix}
    X_1 & X_2 & . & . & .  & X_n \\
\end{bmatrix}$.
$$ \Longrightarrow X v_e = X \frac{-1}{\alpha_e} \sum\limits_{j = 1, j \neq e}^{l} \alpha_j v_{i_j}$$
$$ \Longrightarrow Y_e = \frac{-1}{\alpha_e}  \sum\limits_{j = 1, j \neq e}^{l} \alpha_j Y_j  \Longrightarrow H(Y_e | \bigcup\limits_{j \neq e, j = 1}^{l} Y_j) = 0$$
Therefore, after removing an element from each cycle of the graph, the entropy of the variables defined above remains the same. So the entropy of the graph is the entropy of its spanning tree $T$.
Now take an arbitrary root $u$ of $T$ and its neighbour vertex by the edge $e^*$ is $v$. So $Y_{e^*} = X_u - X_v$ or $Y_{e^*} = X_v - X_u$.
Without loss of generality assume that $Y_{e^*} = X_u - X_v$.  Since $u$ is a root of $T$, it does not appear on the variables of other edges of $T$. So :
$$ H(X_u - X_v | \bigcup\limits_{e \in T , e \neq e^*} Y_e) \geq H(X_u - X_v | \bigcup\limits_{i \neq u} X_i) $$
$$ = H(X_u - X_v | X_v).$$
As we know, if $X_u$ is a uniform random variable on $[k]$, then $X_u - \alpha$ is also a uniform random variable on $[k]$. So, we have $$H(X_u - X_v | X_v) = H(X_u - X_v) 
$$ $$\Longrightarrow H(Y_{e^*} | \bigcup\limits_{e \in T, e \neq e^*} Y_e) = H(X_u - X_v) = H(Y_{e^*})$$
$$ \Longrightarrow H(\bigcup\limits_{e \in T} Y_e) = H(Y_{e^*}) + H(\bigcup\limits_{e \in T, e \neq e^*} Y_e)$$
The above arguments are true if $Y_{e^*} = X_v - X_u$.
By repeating this, we have 
$$ H(\bigcup\limits_{e \in T} Y_e) = \sum\limits_{e \in T} H(Y_e) = |T| \log k = \phi_w(T) = \phi_w(E)$$
\end{proof}

\section*{Discussion}

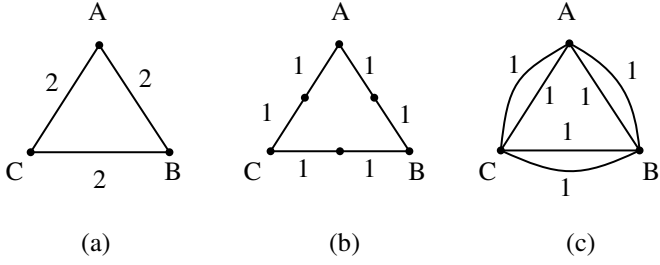
\begin{figure}
    \centering

\tikzset{every picture/.style={line width=0.75pt}} 

\begin{tikzpicture}[x=0.75pt,y=0.75pt,yscale=-1,xscale=1]
\fill [black] (233.6,326.05) circle (2);
\fill [black] (269,380.1) circle (2);
\fill [black] (199,380.1) circle (2);
\fill [black] (78,380.6) circle (2);
\fill [black] (112.6,326.55) circle (2);
\fill [black] (148,380.6) circle (2);
\fill [black] (251.3,353.075) circle (2);
\fill [black] (234,380.1) circle (2);
\fill [black] (216.3,353.075) circle (2);
\fill [black] (349.7, 325.65) circle (2);
\fill [black] (385.1,379.7) circle (2);
\fill [black] (315.1,379.7) circle (2);
\draw   (233.6,326.05) -- (269,380.1) -- (199,380.1) -- cycle ;
\draw   (112.6,326.55) -- (148,380.6) -- (78,380.6) -- cycle ;
\draw   (349.7,325.65) -- (385.1,379.7) -- (315.1,379.7) -- cycle ;
\draw    (349.7,325.65) .. controls (377.2,340.15) and (382.7,355.65) .. (385.1,379.7) ;
\draw    (349.7,325.65) .. controls (322.7,341.15) and (318.7,345.15) .. (315.1,379.7) ;
\draw    (315.1,379.7) .. controls (344.7,393.15) and (354.7,394.15) .. (385.1,379.7) ;

\draw (64,384) node [anchor=north west][inner sep=0.75pt]   [align=left] {C};
\draw (267,384) node [anchor=north west][inner sep=0.75pt]   [align=left] {B};
\draw (385,384) node [anchor=north west][inner sep=0.75pt]   [align=left] {B};
\draw (302.6,384) node [anchor=north west][inner sep=0.75pt]   [align=left] {C};
\draw (184,384) node [anchor=north west][inner sep=0.75pt]   [align=left] {C};
\draw (144,384) node [anchor=north west][inner sep=0.75pt]   [align=left] {B};
\draw (228,303.3) node [anchor=north west][inner sep=0.75pt]   [align=left] {A};
\draw (105.5,303.3) node [anchor=north west][inner sep=0.75pt]   [align=left] {A};
\draw (84,339.3) node [anchor=north west][inner sep=0.75pt]   [align=left] {2};
\draw (317,330) node [anchor=north west][inner sep=0.75pt]   [align=left] {1};
\draw (108,388) node [anchor=north west][inner sep=0.75pt]   [align=left] {2};
\draw (131,337.3) node [anchor=north west][inner sep=0.75pt]   [align=left] {2};
\draw (343.1,391.9) node [anchor=north west][inner sep=0.75pt]   [align=left] {1};
\draw (335,345.9) node [anchor=north west][inner sep=0.75pt]   [align=left] {1};
\draw (344,364) node [anchor=north west][inner sep=0.75pt]   [align=left] {1};
\draw (353.1,346.4) node [anchor=north west][inner sep=0.75pt]   [align=left] {1};
\draw (376.6,332.9) node [anchor=north west][inner sep=0.75pt]   [align=left] {1};
\draw (342.1,303.4) node [anchor=north west][inner sep=0.75pt]   [align=left] {A};
\draw (102,420) node [anchor=north west][inner sep=0.75pt]   [align=left] {(a)};
\draw (227,420) node [anchor=north west][inner sep=0.75pt]   [align=left] {(b)};
\draw (346.6,420) node [anchor=north west][inner sep=0.75pt]   [align=left] {(c)};
\draw (218.5,388.52) node   [align=left] {\begin{minipage}[lt]{9.38pt}\setlength\topsep{0pt}
1
\end{minipage}};
\draw (216.1,336.52) node   [align=left] {\begin{minipage}[lt]{9.38pt}\setlength\topsep{0pt}
1
\end{minipage}};
\draw (252,388.52) node   [align=left] {\begin{minipage}[lt]{9.38pt}\setlength\topsep{0pt}
1
\end{minipage}};
\draw (200.1,360.12) node   [align=left] {\begin{minipage}[lt]{9.38pt}\setlength\topsep{0pt}
1
\end{minipage}};
\draw (270,361.32) node   [align=left] {\begin{minipage}[lt]{9.38pt}\setlength\topsep{0pt}
1
\end{minipage}};
\draw (252,336.52) node   [align=left] {\begin{minipage}[lt]{9.38pt}\setlength\topsep{0pt}
1
\end{minipage}};

\end{tikzpicture}
    \caption{(a): the weights of all edges are $2$ and the entropy of the graph is $4$.
    (b): serial multiplication; the entropy of the graph is $5$.
    (c): parallel multiplication; the entropy of the graph is $2$.}
    \label{fig:fig2}
\end{figure}

When the weight function is constantly equal to $w$ on the singletons, the weighted rank function is identical to the usual rank function multiplied by the constant $w$. This function is shown to be entropic when the underlying matroid is representable over a finite field $\mathbb{F}_q$ and $w=\log (q)$. Similarly, if we evaluate the entropy with the logarithm function in base $q$, then Theorem \ref{thm:main} holds for any representable matroid over the field $\mathbb{F}_q$. The proof is identical to the binary case. For the sake of simplicity, we only consider the binary case in the statement of the theorem.

Another important point is that when we deal with integral weight functions on mathematical objects, it is often the case that one can replace the integral weights by modifying the objects and making all the weights equal to one. A typical example is the case of integral weighs on the edges of a graph. If we replace each edge with a path of length equal to the weight of the edge, then many properties in the graph remains unchanged. For example, this procedure does not change the shortest path between any pair of the original vertices. For our problem, one may naturally expect that when the weight of an element is multiplied by an integer, the same argument as \cite{abbe2019entropic} on a new matroid whose corresponding element is replaced by a sequence of parallel elements will show that still the weighted rank function is entropic. However, in Figure \ref{fig:fig2}, we present an example that shows that this idea does not work.  The idea of replacing an element of the matroid with a series of elements does not work either as shown in Figure \ref{fig:fig2}. This shows that our result indeed uses different ideas as the one in the paper \cite{abbe2019entropic}.

\nocite{*}
\bibliographystyle{plain} 
\bibliography{ref} 

\section*{Appendix}
\begin{proof}[Proof of Proposition \ref{pro:submodular}]
\label{proof:submodularphi}
To show this, we must prove that for every subset $A \subseteq [n]$ and two distinct elements $i,j \notin A$ we have
\begin{equation}
    \phi_w(A \cup i) - \phi_w(A) \geq \phi_w(A \cup i \cup j) - \phi_w(A \cup j)
\end{equation}

Suppose that $\phi_w$ is not submodular. Let $A$ be a minimum size set such that (1) is not valid. That is 
$$ \phi_w(A \cup i) - \phi_w(A) < \phi_w(A \cup i \cup j) - \phi_w(A \cup j)$$
We cliam that $A$ is an independent set. In fact, if $C$ is a circuit in $A$ and $a \in C$ is the minimum weight element of $C$, then by  Lemma \ref{lemma:MWB}, $\phi_w(X) = \phi_w(X \backslash a)$ for $X = A, A \cup i , A \cup j, A \cup i \cup j$.
This shows that $A \backslash a $ also does not satisfy (1) which contradicts the minimality of $A$. Now, let $C_1, C_2$ be the unique circuits of $A \cup i$ and $A \cup j$, respectively.
\\ Let $a, b$ be the elements of $C_2, C_2$ of minimum weight, respectively.\\
Since $A$ is shown to be independent, we can see that $\phi_w(A \cup i) = \phi_w(A) + w_i - w_a$. Similarly, $\phi_w(A \cup j) = \phi_w(A) + w_j - w_b$. \\
To compute $\phi_w(A \cup i \cup j)$, first notice that $A \cup i \cup j$ contains at least two circuits $C_1, C_2$. In order to have an independent set in $A \cup i \cup j$, we must remove at least one element from each $C_i$. On the other hand, we may not eliminate both circuits by removing an element that is shared in both $C_1, C_2$ since the resulting set still contains a circuit.(See proposition \ref{prop:circuit}) Hence, $\phi_w(A \cup i \cup j) \leq \phi_w(A) + w_i + w_j - w_a - w_b$. By substituting these bounds in (1) we get a contradiction.

\end{proof}

\end{document}